\documentclass{article}

\usepackage{a4wide,amssymb,amsmath,amsthm,bbm}

\def\CC{{\mathbb C}}

\def\A{{\mathcal A}}

\def\B{{\mathcal B}}
\def\E{{\mathcal E}}
\def\F{{\mathcal F}}

\newtheorem{theorem}{Theorem}[section]
\newtheorem{corollary}[theorem]{Corollary}
\newtheorem{proposition}[theorem]{Proposition}
\newtheorem{lemma}[theorem]{Lemma}

\begin{document}

\title{A gap for the maximum number
of mutually unbiased bases}

\author{{\bf Mih\'aly Weiner}\footnote{On leave from the Alfr\'ed R\'enyi 
Institute of Mathematics, Budapest.} \footnote{Supported 
in part by the ERC Advanced Grant 227458 OACFT ``Operator Algebras and Conformal 
Field Theory'' and the Momentum fund of the Hungarian Academy of Sciences.}\\ 
Dep.\! of Mathematics, University of Rome ``Tor Vergata''
\\ {\tt{mweiner{@}renyi.hu}}}
\date{}
\maketitle


\begin{abstract}
A collection of (pairwise) mutually unbiased bases (in short: MUB)
in $d>1$ dimensions may consist of at most $d+1$ bases. Such 
``complete'' collections are known to exists in $\CC^d$ when $d$ is a 
power of a prime. However, in general little is known about the maximum 
number $N(d)$ of bases that a collection of MUB in $\CC^d$ can have.

In this work it is proved that a collection of $d$ MUB in
$\CC^d$ can be always completed. Hence $N(d)\neq d$ and when $d>1$ we have
a dichotomy: {\it either} $N(d)=d+1$ (so that there exists a complete
collection of MUB) {\it or} $N(d)\leq d-1$.
In the course of the proof an interesting new characterization is given 
for a linear subspace of $M_d(\CC)$ to be a subalgebra.
\end{abstract}

\section{Introduction}

Two orthonormal bases $\E=(\mathbf e_1, \ldots, \mathbf
e_d)$ and $\F=(\mathbf f_1, \ldots, \mathbf f_d)$ in $\CC^d$ such that 
\begin{equation}
|\langle \mathbf e_k,\mathbf f_j \rangle| = {\rm constant} =
\frac{1}{\sqrt{d}}
\end{equation}
for all $k,j=1,\ldots, d$,
are said to be {\bf mutually unbiased}. A famous question regarding
mutually unbiased bases (MUB) is the following: in a $d$-dimensional
complex space, at most how many orthonormal bases can be given so that 
any two of them are mutually unbiased?

The motivation of the question is coming from quantum information 
theory. MUB are useful in quantum state tomography \cite{wooters}, and 
the known quantum cryptographic protocols also rely on MUB; see for 
example \cite{security}.

Simple arguments show that the maximum number $N(d)$ of orthonormal 
bases in a collection of MUB satisfies the bound $N(d)\leq d+1$ 
for every $d>1$. A collection of $d+1$ MUB is usually referred as a 
{\bf complete collection}. When the dimension $d=p^\alpha$ is a power of 
a prime, such complete collections can be constructed 
\cite{ivanovic,woofie}. However, 
apart from this case, at the moment there is no dimension $d>1$ in which 
the value of $N(d)$ would be known. So already in dimension six the 
problem is open. Nevertheless, numerical and other evidences 
\cite{numerical,kozos} suggests that $N(6)=3$, which is much less than 
$7$ (that we would need for a complete collection.)

It seems that the problem of complete collections of 
MUB is deeply related to that of finite
projective planes (or equivalently: to complete collections of
mutually orthogonal Latin squares); see for example the construction
\cite{Latinconstruction} and the overview \cite{projective}. However,
it has not been proved that either of the two --- namely, the existence
of a finite projective plane of order $d$ and the existence of a
complete collection of MUB in $\CC^d$ ---  would imply the other.

In this respect, the result of the present work can be considered as one
more indication of the connection between the two questions. Here it
will be proved that having a collection of $d$ MUB in $\CC^d$, one
can always find and add one more basis with which it becomes
a complete collection. In general, if a collection is ``missing'' two
bases, it cannot be always completed and the first example for this occurs 
in $d=4$ dimensions; see \cite{2to5}.
This is in perfect similarity with the following. 
A collection of mutually orthogonal Latin squares ``missing''
only one element to be complete can be indeed completed\footnote{This is
well-known to experts of the field \cite{zspsz}, but it 
is difficult to give a good reference. One may say that 
it is subcase of \cite[Theorem 4.3]{bruck}, 
but it is somewhat missleading as 
the proof of this much stronger statement 
is difficult, whereas what we need is 
almost a triviality, 
e.g. in the textbook
\cite{combcourse} it is given as an exercise.}.
In general, a collection of mutually orthogonal $n\times n$ Latin 
squares ``missing'' two elements cannot be always 
completed and the smallest value\footnote{It is evident that for 
$n=1,2,3$ there can be no such example. For $n=4$ finding such 
an example simply means finding a ``bachelor'' $4\times 4$ Latin square;
i.e. one that has no orthogonal mate. The existence of bachelor Latin 
squares of many different sizes were already known to Euler and in 
\cite{bachelorLatin} it is proved that for any $n\geq 4$ there exists a 
bachelor Latin square.}
(and by \cite{2missingLatin} in fact the {\it only} value) of $n$ for which 
such an incomplete collection can be given is $n=4$.

One may have a look at the problem of MUB from several different 
point of views. It may be considered to regard Lie algebra theory 
\cite{liealg}. The original problem, which is formulated in a complex 
space, may be also turned into a real convex geometrical question and 
hence may be investigated with tools of convex geometry \cite{polytop}. 
Often questions about MUB are rephrased in terms of complex Hadamard 
matrices; see for example \cite{h6}. However, for the author of this 
work, the most natural point of view is that of operator algebras (or, 
being in finite dimensions, perhaps better to say: matrix algebras).

There is a natural way to associate a maximal abelian $*$-subalgebra
(in short: a MASA) to an orthonormal basis (ONB). In the context of 
matrix algebras, we consider a system of MASAs instead of a system of 
bases. Mutual unbiasedness is then expressed as a 
natural orthogonality relation (sometimes also called ``quasi-orthogonality''
or ``complementarity of subalgebras''). 
In fact, in the study of matrix algebras one considers systems of
orthogonal subalgebras {\it in general} (that is, systems
consisting of all kind of subalgebras --- not only maximal abelian
ones). For the topic of orthogonal subalgebras and its relation to
mutual unbiasedness see for example \cite{petz,petzkhan,ohno,opsz,pszw} and
\cite{sajat_uj}. Note that apart from the finite dimensional case, 
orthogonal subalgebras were also considered in the context of type 
I$\!$I$_1$ von Neumann algebras; see \cite{popa}.

Suppose $\A_1,\ldots, \A_d,\A_{d+1}$ is a complete collection of 
quasi-orthogonal MASAs in $M_d(\CC)$. Then $\A_{d+1}$ must be
the orthogonal complement of $V:=+_{k=1}^d (\A_k  \cap
\{\mathbbm 1\}^\perp)$. So if we are only given $d$
quasi-orthogonal MASAs, then only at one place
we can possibly find a MASA which is quasi-orthogonal to all of them: at
the orthogonal complement of $V$. All we need to show is that this 
subspace of $M_d(\CC)$ --- which is {\it a priori}
not even an algebra --- is in fact a MASA. This will 
be done by first working out an interesting new characterization
for a linear subspace of $M_d(\CC)$ to be a subalgebra. 

Can we find a (closed, ``elementary'') expression giving the 
``missing basis'' in terms of the others? It is clear where the 
``missing'' MASA is, but to find the corresponding {\it basis} we 
would need to diagonalize the matrices appearing in our MASA. This might 
require to find the roots of certain 
characteristic polynomials. 
So note that it 
might well be that in general in dimensions $d\geq 5$ there is no 
(closed, ``elementary'') expression giving the missing basis.

\section{Preliminaries}

Let $\E=(\mathbf e_1,\ldots,\mathbf 
e_d)$ be an ONB in $\CC^d$, and denote the ortho-projection onto the 
one-dimensional subspace $\CC\mathbf e_j$ by $P_{\mathbf e_j}$ for each 
$j=1,\ldots, d$. Then we may consider 
\begin{equation} \A_\E= {\rm 
Span}\{\, P_{\mathbf e_j}\, |j=1,\ldots, d\}, 
\end{equation} 
that is, 
the subspace of $M_d(\CC)$ spanned linearly by the ortho-projections 
$P_{\mathbf e_j}$ $(j=1,\ldots,d)$. It is a MASA, and actually, if 
$\A\subset M_d(\CC)$ is a MASA, then there exists an ONB $\E$ such that 
$\A=\A_\E$.

There is a natural scalar product on $M_d(\CC)$; the so-called
{\it Hilbert-Schmidt} scalar product, defined by the formula
\begin{equation}
\langle A,B\rangle = {\rm Tr}(A^*B)\;\;\;\;\;\;\;\;\; (A,B\in M_d(\CC)).
\end{equation}
In this sense, if $\A\subset M_d(\CC)$ is a given linear subspace, one
can consider the ortho-projection $E_\A$ onto $\A$. When $\A$ is
actually a $*$-subalgebra containing $\mathbbm 1 \in M_d(\CC)$, then
$E_\A$ is nothing else than the so-called {\it trace-preserving
conditional expectation} onto $\A$. If more in particular $\A =\A_\E$ is
the MASA associated to the ONB $\E$, then an easy check shows that
\begin{equation}\label{expectationontoMASA1}
E_{\A_\E}(X) = \sum_{k=1}^d P_{\mathbf e_k} X P_{\mathbf e_k}
\end{equation}
for all $X\in M_d(\CC)$.

Two MASAs $\A,\B\subset M_d(\CC)$, as subspaces, cannot be orthogonal,
since $\A\cap \B\neq \{0\}$ as $\mathbbm 1\in \A\cap \B$. At most, the
subspaces $\A\cap \{\mathbbm 1\}^\perp$ and $\B\cap \{\mathbbm
1\}^\perp$ can be orthogonal, in which case we say that $\A$ and $\B$
are {\bf orthogonal subalgebras}. A direct consequence of the defintions 
of the Hilbert-Schmidt scalar product and of subalgebra-orthogonality is
that $\A$ and $\B$ are orthogonal subalgebras of $M_d(\CC)$ if and only if for all 
$A\in \A$ and $B\in \B$,
\begin{equation}\label{product-trace}
\tau(A B) = \tau(A)\tau(B),
\end{equation}
where $\tau = \frac{1}{d}{\rm Tr}$ is the normalized trace. 

As is well-known, --- but in any case it can be obtained by simply
substituting $A:=P_{\mathbf e_k}$ and $B:=P_{\mathbf f_j}$
into (\ref{product-trace}) --- two MASAs $\A_\E$ and
$\A_\F$ in $M_d(\CC)$ are orthogonal if and only if $\E$ and $\F$
are mutually unbiased. So the problem of finding a certain number of
MUB is equivalent to finding the same number
of orthogonal MASAs.

The dimension of $\A\cap \{\mathbbm 1\}^\perp$
is ${\rm dim}(\A)-1 = d-1$ for a MASA $\A$, whereas the dimension
of $M_d(\CC)\cap \{\mathbbm 1\}^\perp$ is $d^2-1$. However, if $d>1$, 
then
in a $(d^2-1)$-dimensional space there can be at most
\begin{equation}
\frac{d^2-1}{d-1} = d+1
\end{equation}
pairwise orthogonal, $(d-1)$-dimensional subspaces. So when $d>1$, a
collection of orthogonal MASAs can have at most $d+1$ elements;
this is one of the ways one can obtain the well-known upper bound on
$N(d)$.

We shall finish this section by recalling an important fact about
orthonormal bases in $M_d(\CC)$. Its proof can be found for example in
\cite{werner}; but one could also have a look at \cite[Proposition
1]{ohnopetz}, which is a stronger generalization. However, for
self-containment let us see now the statement together with its proof.

\begin{lemma}\label{citedlemma}
Let $A_1,\ldots , A_{d^2}$ be an ONB in $M_d(\CC)$. Then
$$\sum_{k=1}^{d^2} A_k^* X A_k =  {\rm Tr}(X) \, \mathbbm 1$$
for all $X\in M_d(\CC)$.
\end{lemma}
\begin{proof}
Let $B_1,\ldots, B_{d^2}$ another ONB in $M_d(\CC)$. Then there exist
complex coefficients $\lambda_{k,j}$ $(k,j=1,\ldots, d^2)$ such that
$B_k = \sum_{j} \lambda_{k,j} A_j$. Since a linear map that takes an ONB
into an ONB must be unitary, we have that $\sum_{k=1}^{d^2}
\overline{\lambda_{k,j}}
\lambda_{k,l} = \delta_{j,l}$. Hence
\begin{equation}
\sum_{k=1}^{d^2} B_k^* X B_k =
\sum_{k,j,l=1}^{d^2} (\lambda_{k,j}A_j)^* X (\lambda_{k,l}
A_k) = \sum_{k,j,l=1}^{d^2} \overline{\lambda_{k,j}}\lambda_{k,l}
A_j^* X A_l 
= \sum_{j=1}^{d^2}
A_j^* X A_j
\end{equation}
showing that the sum appearing in the statement is independent of the
chosen ONB. Thus the formula can be verified by an elementary check
using the ONB consisting of ``matrix units''.
\end{proof}

Note that the same argument, together with formula
(\ref{expectationontoMASA1}), shows that if $\A\subset M_d(\CC)$ is a
MASA then for {\it any} ONB
$A_1,\ldots, A_d$ in $\A$ we have that
\begin{equation}\label{expectationontoMASA2}
E_{\A}(X) = \sum_k A_k^* X A_k.
\end{equation}
for all $X\in M_d(\CC)$.

\section{The ``missing'' basis found}

Suppose we are given a collection of $d$ MUB in $\CC^d$. As was 
explained, this gives us $d$ pairwise
orthogonal MASAs in $M_d(\CC)$; let us denote them by
$\A_1,\ldots, \A_d$.

The subspaces $\A_k\cap \{\mathbbm 1\}^\perp$ $(k=1,\ldots, d)$ are
$d-1$ dimensional, orthogonal subspaces. Hence $V:= +_{k=1}^d (\A_k\cap
\{\mathbbm 1\}^\perp)$ is $(d^2-d)$-dimensional, and $V^\perp$ is a
$d$-dimensional subspace in $M_d(\CC)$. Our aim is to prove that
$\B:=V^\perp$ is actually a MASA. However, it is not even clear whether
it is an {\it algebra} (that is, whether it is closed for
multiplication). There are two things though that are rather
evident. First, that $\mathbbm 1\in \B$. Second, that $\B$ is a 
self-adjoint subspace: $X\in \B\, \Leftrightarrow X^*\in\B$. This 
second property follows easily from the fact that it holds for
$\A_1,\ldots \A_d$ and that the restriction of the Hilbert-Schmidt
scalar product onto the real subspace of self-adjoints is real.
\begin{lemma}
\label{subalg_char}
Let $K\subset M_d(\CC)$ be a self-adjoint linear subspace containing 
$\mathbbm 1\in M_d(\CC)$ and let further $E_K$ stand for the ortho-projection 
onto $K$. Then $K$ is a  subalgebra of $M_d(\CC)$ if and only if $E_K$ is 
2-positive. 
\end{lemma}
\begin{proof}
First let us note that $E_K$ automatically preserves the trace: 
\begin{equation}
{\rm Tr}(E_K(X)) = 
\langle \mathbbm 1, E_K(X)\rangle = \langle E_K(\mathbbm 
1),X\rangle = 
\langle \mathbbm 1, X\rangle = {\rm Tr}(X).
\end{equation}
Now if $K$ is a subalgebra of $M_d(\CC)$, then $E_K$ is the trace-preserving 
conditional expectation onto $K$ whose {\it complete} positivity is well-known. 
{\it Vice versa}, if $E_K$ is 2-positive then by \cite[Corollary 2.8]{choi} one has 
the operator-inequality
\begin{equation}
E_K(X^*X) \geq E_K(X^*)E_K(X).
\end{equation}
In particular, if $X\in K$ then $E_K(X^*X) \geq X^*X$ and by 
applying the trace on both sides one further sees that it is actually an equality: 
$E_K(X^*X)= X^*X = E_K(X^*)E_K(X)$. Then by \cite[Theorem 3.1]{choi} it follows 
that $K$ is in the {\it multiplicative domain} of $E_K$. Hence if $X,Y\in K$    
then $XY = E_K(X)E_K(Y) = E_K(XY)\in K$ showing that $K$ is a subalgebra of 
$M_d(\CC)$.
\end{proof}
\begin{lemma}
Let $B_1,\ldots,B_n$ an ONB in $\B$. Then
$E_\B(X)=\sum_k B_k^* X B_k$ for all $X\in M_d(\CC)$, where
$E_\B$ is the ortho-projection onto $\B$.
\end{lemma}
\begin{proof}
Let us fix an ONB $A^{(k)}_1\ldots A^{(k)}_{d-1}$ in
$(\A_k\cap \{\mathbbm 1\}^\perp)$ for each $k=1,\ldots, d$.
Then, on one hand, $A^{(k)}_1\ldots
A^{(k)}_{d-1},\frac{1}{\sqrt{d}}\mathbbm 1$ is an ONB in $\A_k$. On the
other hand, the $d(d-1)$ elements, $A^{(k)}_j$
$(k=1,\ldots,d;\, j=1,\ldots,d-1)$, together with
$B_1,\ldots, B_d$, form an ONB in the full space $M_d(\CC)$. So,
on one hand, by formula (\ref{expectationontoMASA2}) we have that
\begin{equation}
\sum_j (A^{(k)}_j)^* X A^{(k)}_j + \frac{1}{\sqrt{d}}\mathbbm 1 X
\frac{1}{\sqrt{d}} \mathbbm 1 = E_{\A_k}(X),
\end{equation}
implying that
$\sum_j (A^{(k)}_j)^* X A^{(k)}_j = E_{\A_k}(X) - \frac{1}{d}X$. On the
other hand, by Lemma \ref{citedlemma},
\begin{equation}
\sum_n B_n^* X B_n + \sum_{k,j} (A^{(k)}_j)^* X A^{(k)}_j = {\rm Tr}(X)
\mathbbm 1.
\end{equation}
Hence
\begin{equation}
\sum_n B_n^* X B_n = {\rm Tr}(X) \mathbbm 1 -
\sum_{k,j} (A^{(k)}_j)^* X A^{(k)}_j = 
X - \sum_{k=1}^d (E_{\A_k}(X)-\frac{1}{d}{\rm Tr}(X)\mathbbm 1).
\end{equation}
But $\frac{1}{d}{\rm Tr}(X) \mathbbm 1= \langle
\frac{1}{\sqrt{d}}\mathbbm 1, X\rangle \frac{1}{\sqrt{d}}\mathbbm 1 =
E_{\CC\mathbbm 1}(X)$. Thus $E_{\A_k}(X)-\frac{1}{d}{\rm Tr}(X)\mathbbm
1 = E_{\A_k}(X) - E_{\CC\mathbbm 1}(X) =E_{(\A_k\cap\{\mathbbm
1\}^\perp)}(X)$, since $\CC\mathbbm \subset \A_k$. So finaly we obtain
that $\sum_n B_n^* X B_n =$
\begin{equation}
= X - \sum_k  E_{(\A_k\cap\{\mathbbm 1\}^\perp)}(X)
= (id-E_V)(X) = E_{V^\perp}(X) = E_\B(X)
\end{equation}
since $V$ is spanned by the $d$ pairwise orthogonal subspaces
$(\A_k\cap\{\mathbbm 1\}^\perp)$ $(k=1,\ldots, d)$.
\end{proof}
\begin{proposition}
The subspace $\B$ is a MASA.
\end{proposition}
\begin{proof}
By our previous lemma $E_\B$ is completely positive, so by lemma \ref{subalg_char} 
$\B$ is an algebra. On the other hand, if $X'\in\B'$ then 
\begin{equation}
E_\B(X') = \sum_{k=1}^d B_k^* X' B_k = X' \sum_{k=1}^d B_k^*B_k = 
X'E_\B(\mathbbm 1) = X'
\end{equation}
showing that $\B'\subset \B$ and hence that $\B'$ is abelian. Thus 
$\B=(\B')'$ is unitarily equivalent to the subalgebra of all block-diagonal 
matrices of $M_d(\CC)$ for some fixed sequence of block-sizes. However, ${\rm 
dim}(\B)=d$ so the only possibility is that all of these blocks are 1-dimensional 
implying that $\B$ is a MASA.
\end{proof}
\begin{corollary}
Suppose that $\E_1,\ldots, \E_d$ is a collection of MUB in $\CC^d$. Then 
there exists an ONB $\E_{d+1}$ so that $\E_1,\ldots, \E_d, \E_{d+1}$ is 
a complete collection of MUB.
\end{corollary}
\bigskip

\noindent
{\bf \large Acknowledgements.} The author would like to thank prof.\!
D.\! Petz who suggested to consider this problem and T.\! Sz\H{o}nyi
for useful information on Latin squares and their letterature.

\end{document}